\documentclass[12pt]{article}

\usepackage{amsthm,amsmath,amssymb}
\usepackage{hyperref}
\usepackage{enumerate}
\usepackage{mathrsfs}
\usepackage{makeidx}

\usepackage{graphicx}
\usepackage{amsfonts}
\usepackage{color}

\newtheorem{thm}{Theorem}[section]
\newtheorem{cor}[thm]{Corollary}
\newtheorem{lem}[thm]{Lemma}
\newtheorem{conj}[thm]{Conjecture}
\newtheorem{assumption}[thm]{Assumption}

\newtheorem{definition}[thm]{Definition}
\newenvironment{defn}{\begin{definition}\rm}{\end{definition}}
\newtheorem{example}[thm]{Example}
\newenvironment{exam}{\begin{example}\rm}{\end{example}}
\newtheorem{remark}[thm]{Remark}

\makeindex

\title{The Langberg-M\'{e}dard Multiple Unicast Conjecture: Stable $3$-Pair Networks~\thanks{This research is partly supported by a grant from the Research Grants Council of the Hong Kong Special Administrative Region, China (Project No. 17301017) and a grant by the National Natural Science Foundation of China (Project No. 61871343).}}

\author{\small \begin{tabular}{ccc}
Kai Cai & Guangyue Han\\
The University of Hong Kong& The University of Hong Kong\\
email: kcai@hku.hk & email: ghan@hku.hk\\
\end{tabular}}

\date{{\normalsize \today}}

\begin{document} \maketitle

\begin{abstract}
The Langberg-M\'{e}dard multiple unicast conjecture claims that for a strongly reachable $k$-pair network, there exists a multi-flow with rate $(1,1,\dots,1)$. In this paper, we show that the conjecture holds true for {\em stable} $3$-pair networks.
\end{abstract}

\section{Introduction} \label{Introduction}

A {\em $k$-pair network} $\mathcal N=(V, A, S, T)$ consists of a directed acyclic graph (DAG) $D=(V, A)$, a set of $k$ vertices $S=\{s_1, s_2, \dots, s_k\}$ with zero in-degree, called sources (senders) and a set of $k$ vertices $T=\{t_1, t_2, \dots, t_k\}$ with zero out-degree, called sinks (receivers). For convenience only, we assume that any $k$-pair network considered in this paper does not have vertices whose in-degree and out-degree are both equal to $1$. Roughly put, the {\em multiple unicast conjecture}, also known as the {\em Li-Li conjecture}~\cite{Li042}, claims that for any $k$-pair network, if information can be transmitted from all the senders to their corresponding receivers at rate $(r_1, r_2,\dots, r_k)$ via network coding, then it can be transmitted at the same rate via undirected fractional routing. One of the most challenging problems in the theory of network coding~\cite{Yeung06}, this conjecture has been doggedly resisting a series of attacks~\cite{CH15, CH16, CH17, CH18, Harv06, Jain06, Langberg09, Zongpeng12, Xiahou12, Yang14} and is still open to date.

A $k$-pair network $\mathcal{N}$ is said to be {\em fully reachable} if there exists an $s_i$-$r_j$ directed path $P_{s_i, r_j}$ for all $i, j$; and {\em strongly reachable} if, in addition, the paths $P_{s_1, r_j}, P_{s_2, r_j}, \cdots, P_{s_k, r_j}$ are {\em edge-disjoint} for any $j$; and {\em extra strongly reachable} if, furthermore, for any $j$ and all $i \neq k$, $P_{s_i, t_j}$ and $P_{s_k, t_j}$ do not share any vertex other than $t_j$. Throughout the paper, we will reserve the notations $\mathbf{P}_{t_j}$ and $\mathbf{P}$ and define
$$
\mathbf P_{t_j}:=\{P_{s_i, t_j}: i = 1, 2, \dots, k\}, \quad \mathbf{P} := \cup_{j=1}^k \mathbf P_{t_j}.
$$
For notational convenience, we may refer to a path from $\mathbf P_{t_j}$ as a $\mathbf P_{t_j}$-path, or simply a $\mathbf{P}$-path, and moreover, an arc on the path $\mathbf{P}_{t_j}$ as a $\mathbf{P}_{t_j}$-arc. Note that an arc can be simultaneously a $\mathbf{P}_{t_j}$-arc and a $\mathbf{P}_{t_{j'}}$-arc, $j \neq j'$.

The following {\em Langberg-M\'{e}dard multiple unicast conjecture}~\cite{Langberg09}, which deals with strongly reachable $k$-pair networks, is a weaker version of the Li-Li conjecture. Note that for a strongly reachable network, each source is able to multicast at rate 1 to all the receivers, e.g., the classic butterfly network of two-unicast is such a case.
\begin{conj}\label{conj-1}
For any strongly reachable $k$-pair network, there exists a feasible undirected fractional multi-flow with rate $(1,1,\dots,1)$.
\end{conj}
\noindent It turns out that Conjecture~\ref{conj-1} is equivalent to the following conjecture, with ``strongly reachable'' replaced by ``extra strongly reachable''.
\begin{conj}\label{conj-2}
For any extra strongly reachable $k$-pair network, there exists a feasible undirected fractional multi-flow with rate $(1,1,\dots,1)$.
\end{conj}
\noindent To see the equivalence, note that Conjecture~\ref{conj-1} trivially implies Conjecture~\ref{conj-2}, and the reverse direction follows from the fact that a strongly reachable $k$-pair network can be transformed to an extra strongly reachable $k$-pair network with a feasible undirected fractional multi-flow mapped to one with the same rate.

The Langberg-M\'{e}dard multiple unicast conjecture was first proposed in 2009~\cite{Langberg09}. In the same paper, the authors constructed a feasible undirected fractional multi-flow with rate $(1/3, 1/3, \dots, 1/3)$ for a strongly reachable $k$-pair network. Recently, we have improved $1/3$ to $8/9$ for a generic $k$ in~\cite{CH15} and to $11/12$ for $k=3, 4$ in~\cite{CH18}.

A strongly reachable $k$-pair network $\mathcal N$ is said to be {\em stable} if the choice of each $P_{s_i, t_j}$, $i, j=1, 2, \dots, k$, is unique, and {\em unstable} otherwise (see Fig.~\ref{Minimal Networks}); here we remark that $\mathcal{N}$ is stable only if it is extra strongly reachable. In this paper, we will establish Conjecture~\ref{conj-1} for stable $3$-pair networks by establishing Conjecture~\ref{conj-2} for the same family of networks. Our treatment is based on classification of stable $3$-pair networks according to their network topologies. Related work on topological analysis of strongly reachable networks can be found in~\cite{Langberg06} and~\cite{Han09}.

\begin{figure}[htbp]
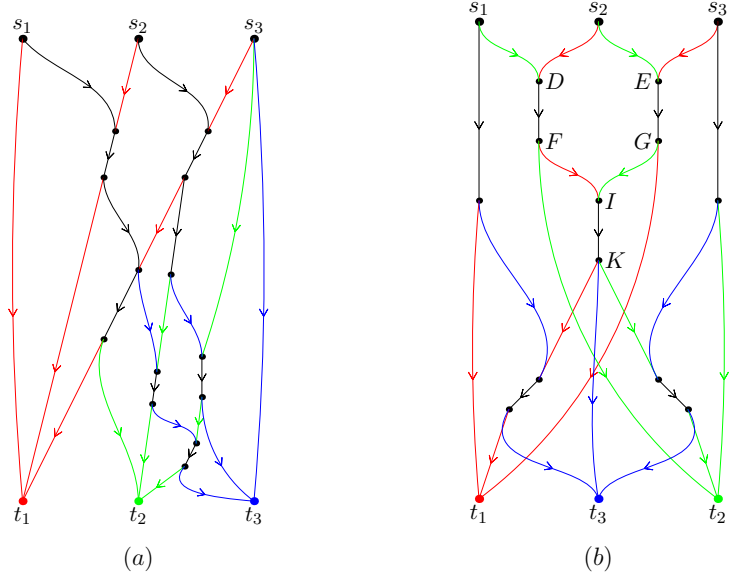

  \centering
  \includegraphics[width=3.5cm]{fig1.1}~~~~~~~~~~~~~~~~~~
  \includegraphics[width=3.5cm]{fig1.2}
  \caption{A stable network $(a)$ and an unstable network $(b)$. In $(b)$, $\mathbf P_{t_2}$ is not unique since $P_{s_2, t_3}$ can be chosen as either $[s_2, D, F, I, K, t_3]$ or $[s_2, E, G, I, K, t_3]$. Here and hereafter, for stable networks, {\bf an arc that belongs to only one $\mathbf P$-path is colored red, green or blue, respectively, depending on the fact that the $\mathbf{P}$-path is a $\mathbf P_{t_1}$-, $\mathbf P_{t_2}$- or $\mathbf{P}_{t_3}$-path; and an arc that belongs to two or more $\mathbf{P}$-paths' is colored black}.}\label{Minimal Networks}
\end{figure}

The rest of paper is organized as follows. In Section 2, we introduce some basic notions, facts and related tools. In Section 3, we characterize stable $k$-pair networks and subsequently show that there exists an efficient algorithm to determine the stability of a given $k$-pair network. In Section4, we investigate the topological structure of stable $3$-pair networks, for which we settle the Langberg-M\'{e}dard conjecture in Section 5. Finally, the paper is concluded in Section 6.

\section{Preliminaries}

Throughout this section, we consider a {\em fully reachable} $k$-pair network $\mathcal{N}$ and adopt all the related notations defined in Section~\ref{Introduction}.

\subsection{Undirected Fractional Multi-Flow}\label{subsection-multi-flow-basics}

For an arc $a=[u,v]\in A$, we call $u$ and $v$ the {\em tail} and the {\em head} of $a$ and denote them by $tail(a)$, $head(a)$, respectively. For any $s, t \in V$, an $s$-$t$ {\em flow}~\footnote{The flow or multi-flow defined for directed graph in this paper, which can be negative, is equivalent to the flow defined in~\cite{Schrijver03} for undirected graphs, which has to be non-negative.} is a function $f: A \rightarrow \mathbb{R}$ satisfying the following {\em flow conservation law:} for any $v \notin \{s, t\}$,
\begin{equation}\label{flow conservation law}
excess_f(v)=0,
\end{equation}
where
\begin{equation}
excess_{f}(v):=\sum_{a\in A: \; head(a)=v} f(a)-\sum_{a\in A: \; tail(a)=v} f(a).
\end{equation}
It is easy to see that $excess_{f}(s)=-excess_{f}(t)$, which is called the {\em value (or rate)} of $f$. We say $f$ is {\it feasible} if $|f(a)| \leq 1$ for all $a\in A$.

An $(s_1, s_2, \dots,s_k)$-$(t_1, t_2, \dots, t_k)$ {\em multi-flow} refers to a set of $k$ flows  $\mathcal{F}=\{f_1, f_2, \dots, f_k\}$, where each $f_{i}$ is an $s_i$-$t_i$ flow. We say $\mathcal{F}$ has {\em rate} $(d_1, d_2, \dots, d_k)$, where $d_i:=excess_{f_i}(s_i)$. For any given $a \in A$, we define $|\mathcal{F}|(a)$ as
\begin{equation}\label{total value in an arc}
|\mathcal{F}|(a):=\sum_{1\leq i\leq k}|f_{i}(a)|.
\end{equation}
And we say $\mathcal{F}$ is {\em feasible} if $|\mathcal{F}|(a) \leq 1$ for all $a \in A$.

\subsection{Routing Solution}

For each $P_{s_i, t_j}$, we define an $s_i$-$t_j$ flow $f_{i,j}$ as follows:
$$
f_{i,j}(a)=\left\{
              \begin{array}{ll}
                1, & \hbox{$a\in P_{s_i, t_j}$,}\\
                0, & \hbox{otherwise.}
              \end{array} \right.
$$
\begin{defn}[Linear Routing Solution]
An $(s_1, s_2, \dots, s_k)$-$(t_1,t_2,\dots, t_k)$ multi-flow $\mathcal{F} = \{f_1, f_2, \dots, f_k\}$ is said to be a {\em routing solution} for $\mathcal{N}$ if it is feasible with rate $(1, 1, \dots, 1)$. A routing solution is called {\em linear} (with respect to $\mathbf{P}$), if, for each feasible $l$,
\begin{equation} \label{coefficient-matrix}
f_{l}=\sum_{i,j=1}^k c^{(l)}_{i,j} f_{i,j},
\end{equation}
where all $c^{(l)}_{i,j} \in \mathbb{R}$, in which case the solution $\mathcal{F}$ can be equivalently represented by its {\em matrix form} $\mathcal C=\left((c^{(1)}_{i,j}), (c^{(2)}_{i,j}), \dots, (c^{(k)}_{i,j})\right)$; otherwise, it is called {\em non-linear}.
\end{defn}

The following theorem~\cite{CH18} is somewhat straightforward.
\begin{thm}\label{basic observation}
An $(s_1, s_2, \dots, s_k)$-$(t_1,t_2,\dots, t_k)$ multi-flow $\mathcal{F} = \{f_1, f_2, \dots, f_k\}$ satisfying (\ref{coefficient-matrix}) has rate $(1,1,\dots,1)$ if and only if all $c^{(l)}_{i,j}$ satisfy
\begin{equation} \label{commodity condition}
\sum_{j=1}^kc^{(l)}_{i,j}=0, \text{ for all}\; i\neq l, \quad \sum_{i=1}^kc^{(l)}_{i,j}=0, \text{ for all}\; j\neq l, \quad \sum_{i=1}^k\sum_{j=1}^kc_{i,j}^{(l)}= 1, \text{ for all}\;l.
\end{equation}
\end{thm}

\begin{figure}[htbp]
  \centering
  \includegraphics[width=4cm]{fig2.1}~~~~~~~~~~
  \includegraphics[width=4cm]{fig2.2}~~~~~~~~~~
  \caption{A Linear routing solution.}\label{2-pair-linear-solution}
\end{figure}

\begin{exam}\label{exam-2-pair}
Consider the $2$-pair network depicted in Fig.~\ref{2-pair-linear-solution} and Fig.~\ref{2-pair-non-linear-solution}. It is easy to check that Fig.~\ref{2-pair-linear-solution} gives a linear routing solution $\mathcal F=(f_1, f_2)$ with the matrix form
$$
\left(\left(
  \begin{array}{cc}
    \frac{3}{4} & \frac{1}{4} \\
    \frac{1}{4} & \frac{-1}{4} \\
  \end{array}
\right),
\left(
  \begin{array}{cc}
    \frac{-1}{4} & \frac{1}{4} \\
    \frac{1}{4} & \frac{3}{4} \\
  \end{array}
\right)
 \right),
$$
i.e., $f_1=\frac{3}{4}f_{1,1}+\frac{1}{4}f_{1,2}-\frac{1}{4}f_{2,2}+\frac{1}{4}f_{2,1}$ and $f_2=\frac{3}{4}f_{2,2}+\frac{1}{4}f_{2,1}-\frac{1}{4}f_{1,1}+\frac{1}{4}f_{1,2}$. Note that
$$
|\mathcal F|(a)=|f_1(a)|+|f_2(a)|=\left\{
 \begin{array}{ll}
 \frac{1}{2}, & \hbox{$a\in\{[s_1,t_2],[s_2,t_1]\}$;} \\
 1, & \hbox{otherwise.}
\end{array}
\right.
$$
On the other hand, it easy to check that Fig.~\ref{2-pair-non-linear-solution} gives a non-linear routing solution $\mathcal{F}=(f_1, f_2)$ with
$$
|\mathcal F|(a)=\left\{
 \begin{array}{ll}
 0, & \hbox{$a=[s_2,t_1]$;} \\
 1, & \hbox{otherwise.}
\end{array}
\right.
$$
\end{exam}

\begin{figure}[htbp]
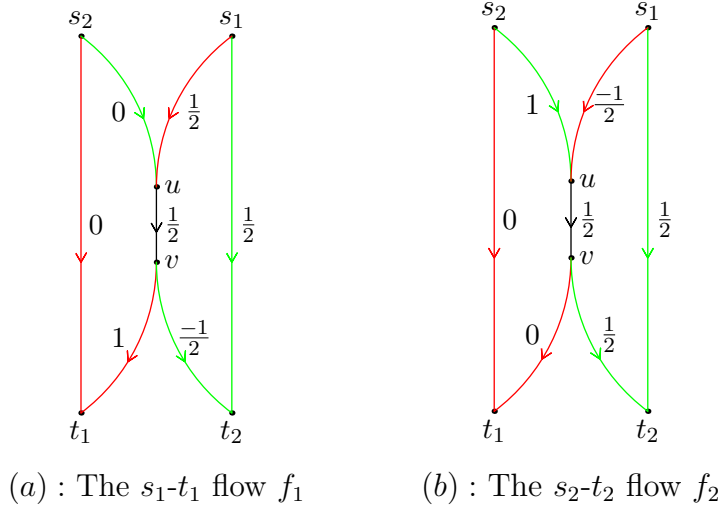

  \centering
  \includegraphics[width=4cm]{fig2.3}~~~~~~~~~~
  \includegraphics[width=4cm]{fig2.4}~~~~~~~~~~
  \caption{A non-linear routing solution.}\label{2-pair-non-linear-solution}
\end{figure}

Using the above language, the  Langberg-M\'{e}dard multiple unicast conjecture says that strongly reachable $k$-pair networks always have routing solutions. Here, we conjecture that it can be further strengthened as follows.
\begin{conj}
Each strongly reachable $k$-pair network has a {\bf linear} routing solution.
\end{conj}

\subsection{$\mathcal{S}_{\mathcal{N}}$ and $g_s(\mathcal{C})$}

Let $[k] = \{1, 2, \dots, k\}$ and define
$$
\mathcal{S}_{\mathcal{N}} = \{ \{(i, j) \in [k] \times [k]: {P}_{s_i, t_j} \mbox{ passes through } a\}: a \in A \};
$$
in other words, each element of $\mathcal{S}_{\mathcal{N}}$ is a set of index pairs corresponding to all $\mathbf{P}$-paths that pass through a given arc in $\mathcal{N}$. Note that, if $\mathcal N$ is a strongly reachable network, then for any feasible $j$, each arc is passed by at most one of the paths $P_{s_1, t_j}, P_{s_2, t_j}, \dots, P_{s_k, t_j}$, and hence $\mathcal S_{\mathcal N}\subseteq\mathcal S_k$, where
$$
\mathcal{S}_k := \{\{(i_1,j_1),\dots (i_r,j_r)\}\subseteq [k]\times[k] \| \, j_1 < j_2 < \dots < j_r, 1 \leq r \leq k\}.
$$

Now, for a tuple of $k\times k$ matrices $\mathcal C=((c^{(1)}_{i,j}), (c^{(2)}_{i,j}), \dots, (c^{(k)}_{i,j}))$ satisfying (\ref{commodity condition}), given $s \in \mathcal{S}_{\mathcal{N}}$ and $l \in [k]$, we define
$$
g^{(l)}_{s}(\mathcal C):=\underset{(i,j)\in s}{\sum}c^{(l)}_{i,j},
$$
and furthermore,
\begin{equation} \label{g-smalls}
g_{s}(\mathcal C):=\sum_{l=1}^k|g^{(l)}_{s}(\mathcal C)|.
\end{equation}
The following theorem, whose proof is straightforward and thus omitted, will be used as a key tool to establish our results.
\begin{thm}\label{thm-comput}
$\mathcal C$ is a linear routing solution of $\mathcal N$ if $g_{s}(\mathcal C)\leq 1$ for any $s\in \mathcal S_{\mathcal N}$.
\end{thm}

For $s=\{(i_1, j_1), (i_2, j_2), \dots, (i_{\alpha(s)}, j_{\alpha(s)})\}\in\mathcal S_{\mathcal N}$,
we define the following multi-set:
$$
Ind_s:=\{i_1, j_1, i_2, j_2, \dots, i_{\alpha(s)}, j_{\alpha(s)}\},
$$
where $\alpha(s)$ denotes the size of $s$. And for any $l=1, 2, \dots, k$, denote by $m_{Ind_s}(l)$ the multiplicity of $l$ in $Ind_s$ (if $l\notin Ind_s$, then $m_{Ind_s}(l)=0$). An element $(i,j)\in s$ is said to be {\em diagonal} if $i=j$, otherwise {\em non-diagonal}. We use $\gamma(s)$ to denote the number of diagonal elements in $s$. For a quick example, consider $s=\{(1,1),(2,2),(1,3),(3,4),(1,6)\}\subseteq[6]\times[6]$. Then, $Ind_s=\{1,1,2,2,1,3,3,4,1,6\}$, $m_{Ind_s}(1)=4$, $m_{Ind_s}(2)=m_{Ind_s}(3)=2$, $m_{Ind_s}(4)=m_{Ind_s}(6)=1$, $m_{Ind_s}(5)=0$, $\alpha(s)=5$ and $\gamma(s)=2$.

\section{Characterization of Stable Networks} \label{section-solutions}

In this section, unless specified otherwise, we assume that $\mathcal{N}$ is an {\em extra strongly reachable} $k$-pair network.

\begin{defn}[Residual Network \cite{Langberg06}]\label{def-residual network}
For $j=1, 2, \dots, k$, the $j$-th {\em residual network} $\mathcal N_j$ is formed from $\mathcal N$ by reversing the directions of all its $\mathbf{P}_{t_j}$-arcs (that may be simultaneously $\mathbf{P}_{t_{j'}}$-arcs for some $j' \neq j$).
\end{defn}
Note that in spite of the acyclicity of $\mathcal N$, there may exist directed cycles in $\mathcal N_j$, and such a directed cycle must contain at least one reversed $\mathbf P_{t_j}$-arc.
\begin{defn}[Regular Cycle]
A directed cycle $C$ of $\mathcal N_j$ is called {\em regular}, if $C$ has no isolated vertex of $\mathbf P_{t_j}$, otherwise it is called {\em singular}.
\end{defn}

\begin{defn}[Semi-Cycle \cite{Han09}]\label{def-semi-cycle}
A $\mathbf P_{t_j}$-{\it semi-cycle} of $\mathcal N$ is formed from a regular cycle of $\mathcal N_j$ by reversing the directions of all its $\mathbf P_{t_j}$-arcs.
\end{defn}

Obviously, there is a one-to-one correspondence from the set of all the $\mathbf P_{t_j}$-semi-cycles in $\mathcal N$ to the set of all the regular cycles of the $j$-th residual network $\mathcal N_j$.

\begin{figure}[htbp]
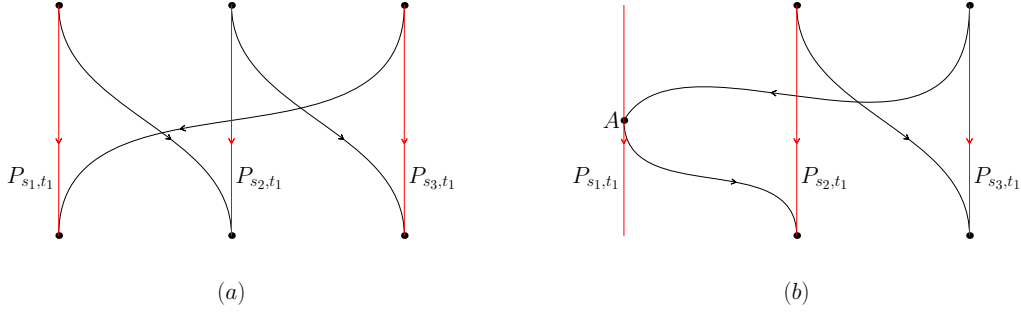

  \centering
  \includegraphics[width=6cm]{fig4.11}~~~~~~~~~~
  \includegraphics[width=6cm]{fig4.12}~~~~~~~~~~
  \caption{If we reverse the directions of all the $\mathbf P_{t_1}$-paths, then both $(a)$ and $(b)$ will give rise to directed cycles of $\mathcal N_1$. Note that the one from $(a)$ is regular, whereas the one from $(b)$ is singular since it contains an isolated vertex $A$ on the path $P_{s_1, t_1}$. And by definition, (a) is a $\mathbf{P}_{t_1}$-semi-cycle.}\label{fig-reg-cycle}
\end{figure}

\begin{defn}[Crossing]\label{def-crossing}
A $\mathbf P_{t_j}$-{\it crossing} of $\mathcal N$ is formed from a $\mathbf P_{t_j}$-semi-cycle of $\mathcal N$ by removing all the $\mathbf P_{t_j}$-arcs.
\end{defn}

For example, consider the network $\mathcal N$ depicted in $(b)$ of Fig.~\ref{Minimal Networks}. While the choices of $P_{s_1, t_3}$ and $P_{s_3, t_3}$ are both unique, there are two choices for $P_{s_2, t_3}$: $P^{(1)}_{s_2, t_3}=[s_2, D, F, I, K, t_3]$ and $P^{(2)}_{s_2, t_3}=[s_2, E, G, I, K, t_3]$, which give rise to two choices of $\mathbf{P}_{t_3}$: $\mathbf P^{(1)}_{t_3}=\{P_{s_1, t_3}, P^{(1)}_{s_2, t_3}, P_{s_3, t_3}\}$ and $\mathbf P^{(2)}_{t_3}=\{P_{s_1, t_3}, P^{(2)}_{s_2, t_3}, P_{s_3, t_3}\}$. By definition, $\{[s_2, D, F, I], [s_2, E, G, I]\}$ is a $\mathbf P^{(1)}_{t_3}$-semi-cycle and also a $\mathbf P^{(2)}_{t_3}$-semi-cycle of $\mathcal N$; $[s_2, E, G, I]$ is a $\mathbf P^1_{t_3}$-crossing and $[s_2, D, F, I]$ is a $\mathbf P^{(2)}_{t_3}$-crossing. If we reverse the direction of $\mathbf P^{(1)}_{t_3}$, then $[s_2, E, G, I, F, D, s_2]$ is a cycle in the corresponding residual network, and if we choose to reverse the direction of $\mathbf P^{(2)}_{t_3}$, then $[s_2, D, F, I, G, E, s_2]$ is a cycle in the corresponding residual network.

We are now ready to state the following theorem, which give characterizations of stable networks.
\begin{thm}\label{thm-main}
For any extra strongly reachable $k$-pair network $\mathcal N$, the following statements are all equivalent.
\begin{description}
  \item[$1)$] $\mathcal N$ is stable.
  \item[$2)$] $\mathcal N$ has no $\mathbf P_{t_j}$-semi-cycle, $j=1, 2, \dots, k$.
  \item[$3)$] $\mathcal N$ has no $\mathbf P_{t_j}$-crossing, $j=1, 2, \dots, k$.
  \item[$4)$] None of $\mathcal N_1, \mathcal N_2, \dots, \mathcal N_k$ has a regular directed cycle.
\end{description}
\end{thm}

\begin{proof}
We will only establish the equivalence between $1)$ and $2)$, which is the only non-trivial part of the proof.

$ 2) \to 1)$: Suppose $\mathcal N$ is unstable. Then, for some $j$, there exist two choices for $\mathbf{P}_{t_j}$: $\mathbf P^{(1)}_{t_j}=\{P^{(1)}_{s_1, t_j}, P^{(1)}_{s_2, t_j}, \dots, P^{(1)}_{s_k, t_j}\}$ and $\mathbf P^{(2)}_{t_j}=\{P^{(2)}_{s_1, t_j}, P^{(2)}_{s_2, t_j}, \dots, P^{(2)}_{s_k, t_j}\}$.  Let $C_1$ be the subnetwork of $\mathcal N$ induced on $\mathbf P^{(1)}_{t_j}\cup \mathbf P^{(2)}_{t_j}$ after removing all the vertices whose in-degree and out-degree are both $1$. Then, for each arc $a$ of $C_1$, there are three cases: $(1)$ $a$ only belongs to a path of $\mathbf P^{(1)}_{t_j}$; $(2)$ $a$ only belongs to a path of $\mathbf P^{(2)}_{t_j}$; $(3)$ $a$ belongs to both a path of $\mathbf P^{(1)}_{t_j}$ and a path of $\mathbf P^{(1)}_{t_j}$. Let $C$ be the digraph induced on the arcs of Cases $(1)$ and $(2)$ after reversing the directions of the arcs of Case $(1)$. It is easy to see that for each vertex $v$ of $C$, the in-degree of $v$ equals the out-degree of $v$. Thus, $C$ is an Eulerian directed graph and hence composed of arc-disjoint directed cycles, which corresponds to a $\mathbf P^{(1)}_{t_j}$-semi-cycle by definition.

$ 1) \to 2)$: Suppose that there exists a $\mathbf P_{t_j}$-semi-cycle $C$, and let $C'$ be the corresponding $\mathbf P_{t_j}$-crossing. Then it is easy to see that $(\mathbf P_{t_j}\setminus C)\cup C'$ is an alternative choice of $\mathbf{P}_{t_j}$, which means that $\mathcal N$ is not stable.
\end{proof}

We would like to add that one can efficiently check that if a given $k$-pair network $\mathcal{N}$ is extra strongly reachable by applying the Ford-Fulkerson algorithm to a set of $k$ directed graphs $D_i$, $1\leq i\leq k$, constructed below:
\begin{itemize}
  \item Add a vertex $s$ as the source node and add an arc $[s, s_j]$ for each $s_j\in S$, $1\leq j\leq k$.
  \item Split each vertex $v\in V\setminus \{s, t_i\}$ into two vertices $v_{in}$ and $v_{out}$ and add an arc $[v_{in}, v_{out}]$. Accordingly, replace arcs $[s, u]$, $[u, v]$ and $[v, t_i]$ by $[s, u_{in}]$, $[u_{out}, v_{in}]$ and $[v_{out}, t_i]$, respectively.
\end{itemize}
It is easy to see that the maximal flow from $s$ to $t_i$ is the number of vertex-disjoint $\mathbf P_{t_i}$-paths, and moreover, if it equals $k$ for all $D_i$, then $\mathcal N$ is extra strongly reachable. Furthermore, it is widely known \cite{CLRS09} that the depth-first search (DFS) algorithm can be used to detect directed cycles in a directed graph, which can be slightly modified \footnote{To see this, whenever the DFS visits a vertex on a $\mathbf{P}_{t_j}$-path from a vertex outside of the $\mathbf{P}_{t_j}$-path, it goes along the $\mathbf{P}_{t_j}$-arc for the next step's visit.} to detect {\bf regular cycles} in a residual network $\mathcal N_j$. To sum up, the equivalence between $1)$ and $4)$ of Theorem~\ref{thm-main}, together with the Ford-Fulkerson algorithm and the DFS algorithm, can be used to efficiently check the stability of a $k$-pair network.

\section{Stable $3$-pair Networks}

In this section, unless specified otherwise, we assume that $\mathcal{N}$ is a stable $3$-pair network. For the sake of convenience, a $\mathbf P_{t_1}$-path, $\mathbf P_{t_2}$-path or $\mathbf P_{t_3}$-path may be referred to as a red path, green path or blue path, respectively. For each feasible $i$, we may use shorthand notations $r_i$, $g_i$ and $b_i$ for paths $P_{s_i, t_1}$, $P_{s_i, t_2}$ and $P_{s_i, t_3}$, respectively. Similarly, we call a $\mathbf P_{t_1}$-crossing, $\mathbf P_{t_2}$-crossing and $\mathbf P_{t_3}$-crossing as a $r$-crossing, $g$-crossing and $b$-crossing, respectively.



\begin{defn}[Longest Common Segment (l.c.s.)]
For any directed paths $p_1, p_2, \dots, p_r$ in $\mathcal{N}$, a longest common segment of $\{p_1, p_2, \dots, p_r\}$, henceforth abbreviated as a $\{p_1, p_2, \dots, p_r\}$-l.c.s., is a segment common to all $p_i$ and any segment properly containing it is not common to all $p_i$.
\end{defn}

\begin{figure}[htbp]
  \centering
  \includegraphics[width=2.5cm]{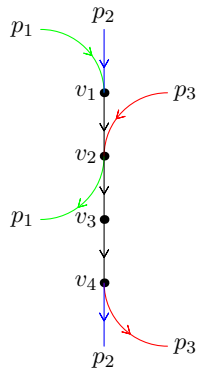}~~~~~~~~~~
  \caption{Longest common segments by paths $p_1$, $p_2$ and $p_3$}\label{merging}
\end{figure}

For example, In Fig.~\ref{merging}, there are three paths $p_1, p_2, p_3$ represented using distinct colors. It is easy to see that $[v_1,v_2]$ is a $\{p_1,p_2\}$-l.c.s.; $[v_2,v_3,v_4]$ is a $\{p_2,p_3\}$-l.c.s.; vertex $v_2$ is a $\{p_1,p_2,p_3\}$-l.c.s. and also a $\{p_1,p_3\}$-l.c.s. On the other hand, $[v_2,v_3]$ is not a $\{p_2,p_3\}$-l.c.s. since $[v_2, v_3, v_4]$, which properly contains $[v_2,v_3]$, is common to both $p_2$ and $p_3$.

\subsection{$\mathcal{N}_{t_i, t_j}$}

For $i \neq j$, we will use $\mathcal{N}_{t_i, t_j}$ to denote the subnetwork of $\mathcal{N}$ induced on all $\mathbf{P}_{t_i}$-paths and $\mathbf{P}_{t_j}$-paths. In this section, we will characterize the topology of $\mathcal{N}_{t_i, t_j}$, and without loss of generality, we will only consider $\mathcal{N}_{t_1, t_2}$. For any path $p\in\mathbf P_{t_1}\cup\mathbf P_{t_2}=\{r_1, r_2, r_3, g_1, g_2, g_3\}$, let $\ell(p)$ denote the number of $\{r_i, g_j\}$-l.c.s.'s ($1\leq i,j\leq 3$) on $p$ and we order all such l.c.s.'s by $p(1) < p(2) < \dots < p(\ell(p))$, where by $p(i)<p(i+1)$, we mean $head(p(i))<tail(p(i+1))$ according to the topological order of the vertices/arcs of a DAG; and we will use $p(i, i+1)$ to denote the path segment of $p$ form $head(p(i))$ to $tail(p(i+1))$. Note that $r_j(1)=g_j(1)$ since $r_j$ and $g_j$ share the same source $s_j$ for all feasible $j$. We first give a simple yet very useful lemma.
\begin{lem}\label{lem-basic-crossing}
For $p, q \in \{r_1, r_2, r_3, g_1, g_2, g_3\}$ and $l = 1, 2, \dots, \ell(p)-1$, if $p(l) \subseteq q$, then $p(l+1) \not \subseteq q$.
\end{lem}

\begin{proof}
Without loss of generality, we suppose $q\in\{r_1,r_2,r_3\}$. Clearly, if $p(l), p(l+1)\subseteq q$, then $p(l, l+1)$ forms a $r$-crossing, which contradict the assumption that the network is stable.
\end{proof}

The following lemma is a key tool in this paper.
\begin{lem}\label{lem-crossing}
There exist $1\leq i\neq j\leq 3$ such that $\ell(r_i) = \ell(g_j) =1$.
\end{lem}

\begin{proof}
${\mathbf(1)}$ We first prove that there exists a green path $g_j$ such that $\ell(g_j) = 1$. To this end, note that if for all $1\leq i\leq 3$, $\ell(g_i)=1$, then the desired result is obviously true. Hence, we suppose, without loss of generality, that $\ell(g_1) > 1$. Clearly, by Lemma~\ref{lem-basic-crossing}, $g_1(2) \not \subseteq r_1$. Thus, we further assume in the following that $g_1(2) \subseteq r_2$ (See Fig.~\ref{proof-lemma-1}(a)).

Now, we consider $g_2$. If $\ell(g_2)=1$, then we are done. So we suppose in the following that $\ell(g_2) > 1$. Then, by Lemma~\ref{lem-basic-crossing}, we deduce that $g_2(2) \not \subseteq r_2$. We also have that $g_2(2) \not \subseteq r_1$ since otherwise $g_1(1, 2)$ and $g_2(1, 2)$ form a $r$-crossing. So, we have $g_2(2)\subseteq r_3$ (See Fig.~\ref{proof-lemma-1}(b)).

Now, consider $g_3$ and suppose, by way of contradiction, that $\ell(g_3) > 1$. Then, we have $(1)$ $g_3(2) \not \subseteq r_3$ (by Lemma \ref{lem-basic-crossing}); $(2)$ $g_3(2) \not \subseteq r_2$ since otherwise $g_2(1, 2)$ and $g_3(1, 2)$ form a $r$-crossing; and $(3)$ $g_3(2)\not \subseteq r_1$ since otherwise $g_1(1, 2)$, $g_2(1, 2)$ and $g_3(1, 2)$ form a $r$-crossing. Hence, we obtain a contradiction to the existence of $g_3(2)$ and thus deduce that $\ell(g_3) = 1$, completing the proof of $(1)$.

${\mathbf(2)}$ By considering the red paths in the parallel manner, we can find a red path, say, $r_i$, such that $\ell(r_i) = 1$.

${\mathbf(3)}$ We now prove $i\neq j$ by contradiction. Without loss of generality, we suppose $i=j=1$, i.e., $\ell(r_1) = \ell(g_1) = 1$. Note that if $\ell(g_2) =1$, then we are done. Hence, we suppose in the following that $\ell(g_2) > 1$. Clearly, $g_2(2) \not \subseteq r_2$ by Lemma~\ref{lem-basic-crossing} and $g_2(2)\not \subseteq r_1$ since $\ell(r_1) = 1$. Hence, $g_2(2) \subseteq r_3$.

Now, consider $g_3$. If $\ell(g_3) = 1$, then we are done since $\ell(r_1)=\ell(g_3)=1$. So, we suppose $\ell(g_3) > 1$. Clearly, $g_3(2) \not \subseteq r_3$ by Lemma~\ref{lem-basic-crossing}, and $g_3(2)\not \subseteq r_2$ since otherwise $g_2(1, 2)$ and $g_3(1, 2)$ form a $r$-crossing. Note that $\ell(r_1) = 1$, we have $g_3(2) \not \subseteq r_1$, which implies that $\ell(g_3) = 1$, completing the proof of the lemma.
\end{proof}

\begin{figure}[htbp]
  \centering
  \includegraphics[width=4cm]{fig1.6}~~~~~~~~~~~~~~~~~~~~~~~
  \includegraphics[width=4cm]{fig1.7}
  \caption{Proof of Lemma \ref{lem-crossing}.}\label{proof-lemma-1}
\end{figure}

A careful examination of the above proof, in particular, Step $(3)$ thereof, reveals that it actually yields a stronger result:
\begin{cor}\label{lem-crossing-2}
If there exists a feasible $i$ such that $\ell(r_i) = 1$ (resp. $\ell(g_i) =1$), then there exists a feasible $j$ such that $j \neq i$ and $\ell(g_j) = 1$ (resp. $\ell(r_j) = 1$).
\end{cor}

\begin{defn}[Non-degenerated $\mathcal{N}_{t_1, t_2}$]
We say $\mathcal{N}_{t_1, t_2}$ is {\em non-degenerated} if there uniquely exist distinct $i,j$ such that $\ell(r_i) = \ell(g_j) = 1$, otherwise {\em degenerated}.
\end{defn}

The following corollary lists all possible topologies of a degenerated $\mathcal{N}_{t_1, t_2}$.

\begin{thm}\label{lem-degenerated-config}
A degenerated $\mathcal{N}_{t_1, t_2}$ is equivalent to $(a)$, $(b)$, or $(c)$ of Fig.~\ref{fig-deg-config} in the sense that the two are isomorphic if each l.c.s. is treated as a single vertex.
\end{thm}

\begin{proof}
We will have to deal with the following two cases:
\begin{description}
\item[$1)$] there exists $i$ such that $\ell(r_i) = \ell(g_i) = 1$. In this case, by Corollary~\ref{lem-crossing-2}, we have the following subcases:
\begin{description}
  \item[$1.1)$] There exists $j \neq i$ such that $\ell(g_j) = 1$; $\ell(r_j) = 1$. In this case, it is easy to see that there exists $l$ distinct from both $i$ and $j$ such that $\ell(r_l) = \ell(g_l) =1$ as shown in $(a)$ of Fig.~\ref{fig-deg-config}.
  \item[$1.2)$] There exist $j \neq i$ and $l \neq i$ such that $\ell(r_j) = 1$; $\ell(g_l) = 1$. In this case, if $\ell(g_j) = 1$, we have Case $1.1)$; otherwise, we have $g_j(2)=r_l(2)$ as shown in $(b)$ of Fig.~\ref{fig-deg-config}.
\end{description}
\item[$2)$] there exist distinct $i,j,l$ such that $\ell(r_i) = \ell(r_j) = \ell(g_l)=1$. In this case, we have either $r_l(2)=g_j(2)$; $r_l(3)=g_i(2)$ as shown in $(c)$ of Fig.~\ref{fig-deg-config} or $r_l(2)=g_i(2)$; $r_l(3)=g_j(2)$ resulting a network equivalent to $(c)$.
\end{description}
The proof is complete by combining all the discussions above.
\end{proof}

\begin{figure}[htbp]
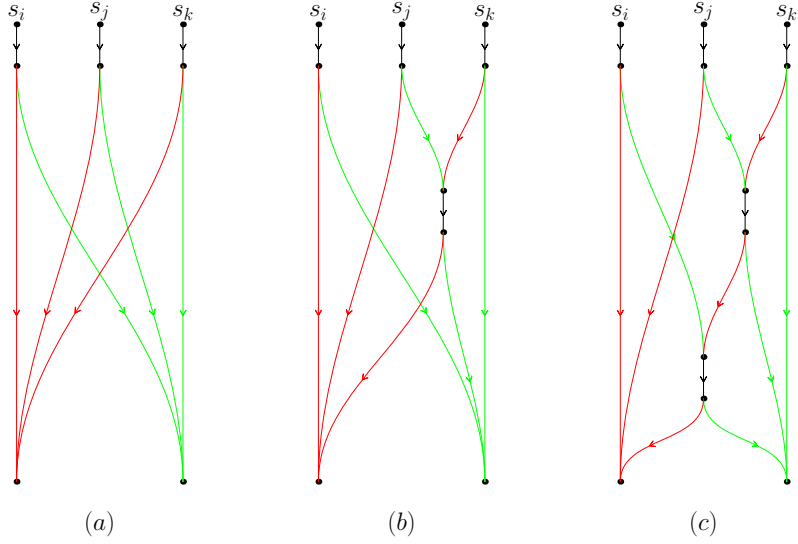

  \centering
  \includegraphics[width=2.5cm]{fig3.6}~~~~~~~~~~
  \includegraphics[width=2.5cm]{fig3.7}~~~~~~~~~~
  \includegraphics[width=2.5cm]{fig3.8}~~~~~~~~~~
  \caption{Possible cases of a degenerated $\mathcal{N}_{t_1, t_2}$.}\label{fig-deg-config}
\end{figure}

\begin{thm}\label{thm-max-config}
A non-degenerated $\mathcal{N}_{t_1, t_2}$ is equivalent to one of the five networks as shown in Fig.~\ref{Non-dege-Config} in the sense that the two are isomorphic if each l.c.s. is treated as a single vertex.
\end{thm}

\begin{proof}
For a non-degenerated $\mathcal{N}_{t_1, t_2}$, we suppose that $\ell(r_i) = \ell(g_l) = 1$ and consider all the possible l.c.s. of $\{r_j, r_l\}$ and $\{g_i, g_j\}$. Recalling that $r_l(1)=g_l(1)$, for all feasible $l$, we start our argument by considering $g_i(2)$ and $g_j(2)$. By Lemma~\ref{lem-basic-crossing}, we have the following two cases:
\begin{description}
\item[$1)$] $g_i(2) \subseteq r_j$ and $g_j(2) \subseteq r_l$. In this case, by Lemma~\ref{lem-basic-crossing}, we infer that $r_j(2) \not \subseteq g_j$ and hence $r_j(2) \subseteq g_i$, which implies that $g_i(2)=r_j(2)$ (due to the acyclicity of $\mathcal{N}$). We then further consider the following two subcases:
\begin{description}
  \item[$1.1)$] $\ell(g_i) = 2$.
  \item[$1.2)$] $\ell(g_i) > 2$.
\end{description}
In Case $1.1)$, it is easy to see that $g_j(2)=r_l(2)$, which leads to the following two subcases:
\begin{description}
  \item[$1.1.1)$] $\ell(g_j) = 2$. In this case, $\mathcal{N}_{t_1, t_2}$ is equivalent to (1.1.1) of Fig.~\ref{Non-dege-Config}.
  \item[$1.1.2)$] $\ell(g_j) \geq 3$. In this case, $g_j(3) \subseteq r_j$. By Lemma~\ref{lem-basic-crossing} and the acylicity of the network, we have $g_j(3)=r_j(3)$. We declare that $\ell(g_j) = 3$ since if otherwise $g_j(4) \subseteq r_l$, again by Lemma~\ref{lem-basic-crossing} and the acyclicity of the network, we have $g_j(4)=r_l(3)$ and hence $r_l(2), r_l(3)\subseteq g_j$, which contradicts Lemma~\ref{lem-basic-crossing}. Hence, in this case, $\mathcal{N}_{t_1, t_2}$ is equivalent to (1.1.2) of Fig.~\ref{Non-dege-Config}.
\end{description}
In Case $1.2)$, since $g_i(2) \subseteq r_j$, we have $g_i(3) \subseteq r_l$. Since $g_j(2), g_i(3) \subseteq r_l$, we have to deal with the following two subcases:
\begin{description}
  \item[$1.2.1)$] $g_j(2)=r_l(2)$ and $g_i(3)=r_l(3)$. In this case, if $\ell(g_j) \geq 3$, then by Lemma~\ref{lem-basic-crossing}, $g_j(3) \subseteq r_j$, which however would imply $g_j(2, 3)$ and $g_i(2, 3)$ form a $r$-crossing. Hence, we have $\ell(g_j) = 2$. Now, if $\ell(g_i) \geq 4$, then by Lemma~\ref{lem-basic-crossing}, $g_i(4) \subseteq r_j$, which further implies $g_i(4)=r_j(3)$. Hence, $r_j(2), r_j(3) \subseteq g_i$, which contradicts Lemma~\ref{lem-basic-crossing}. Hence $\ell(g_j) =2 , \ell(g_i) =3$, and $\mathcal{N}_{t_1, t_2}$ is equivalent to (1.2.1) of Fig.~\ref{Non-dege-Config}.
  \item[$1.2.2)$] $g_i(3)=r_l(2)$ and $g_j(2)=r_l(3)$. In this case, if $\ell(g_i) \geq 4$, then $g_i(4) \subseteq r_j$ and hence $g_i(3, 4)$ and $g_j(1, 2)$ form a $r$-crossing, which contradicts the stability of the network. Thus, $\ell(g_i) =3 $ and we have to consider the following two subcases:
\end{description}
\begin{description}
  \item[$1.2.2.1)$] $\ell(g_j) =2$. In this case, we conclude that $\mathcal{N}_{t_1, t_2}$ is equivalent to (1.2.2.1) of Fig.~\ref{Non-dege-Config}.
  \item[$1.2.2.2)$] $\ell(g_j) \geq 3$. In this case, by Lemma~\ref{lem-basic-crossing}, we have $g_j(3) \subseteq r_j$, which further implies $g_j(3)=r_j(3)$. Now, if $\ell(g_j) \geq 4$, then by Lemma~\ref{lem-basic-crossing}, $g_j(4) \subseteq r_l$, which further implies $g_j(4)=r_l(4)$ and hence $r_l(3),r_l(4) \subseteq g_j$, which contradicts Lemma~\ref{lem-basic-crossing}. Hence $\ell(g_j) =3$ and we conclude that $\mathcal{N}_{t_1, t_2}$ is equivalent to (1.2.2.2) of Fig.~\ref{Non-dege-Config}.
\end{description}

\item[$2)$] $g_i(2) \subseteq r_l$ and $g_j(2) \subseteq r_l$. In this case, without loss of generality, we assume $g_j(2)=r_l(2)$ and $g_i(2)=r_l(3)$ (since otherwise, we can relabel $s_i$, $s_j$ as $s_j$, $s_i$, respectively). By Lemma~\ref{lem-basic-crossing}, we have $g_i(3) \subseteq r_j$ and $g_j(3) \subseteq r_j$. Then, there are two cases:
\begin{description}
  \item[$2.1)$] $g_j(3)=r_j(2)$;
  \item[$2.2)$] $g_i(3)=r_j(2)$.
\end{description}
It is easy to see that $2.1)$ is impossible since otherwise $r_j(1), r_j(2) \subseteq g_j$, which contradicts Lemma~\ref{lem-basic-crossing}. Hence, we have $r_j(2) \subseteq g_i$ and $r_l(2) \subseteq g_j$. By switching the colors of the paths and relabeling sources $s_i$, $s_l$ as $s_l$, $s_i$, respectively, we will reach Case $1)$, which has been dealt with before.
\end{description}

\end{proof}
\begin{figure}[htbp]
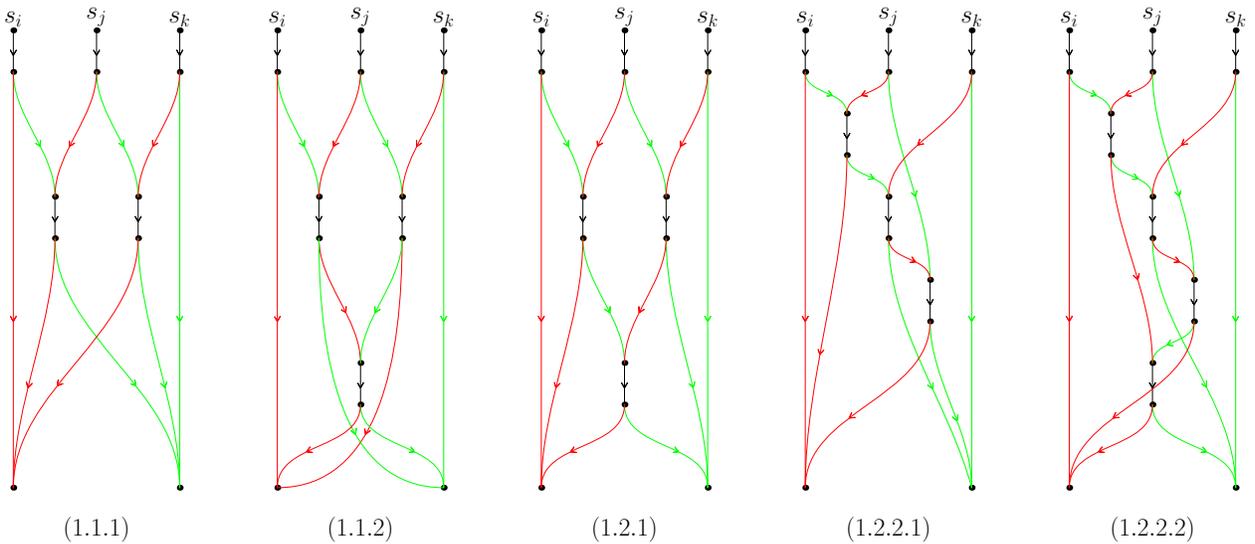

  \centering
  \includegraphics[width=2.5cm]{fig3.1}~~~~~~~~~~
  \includegraphics[width=2.5cm]{fig3.2}~~~~~~~~~~
  \includegraphics[width=2.5cm]{fig3.3}~~~~~~~~~~
  \includegraphics[width=2.5cm]{fig3.4}~~~~~~~~~~
  \includegraphics[width=2.5cm]{fig3.5}~~~~~~~~~~
  \caption{Possible cases of a non-degenerated $\mathcal{N}_{t_1, t_2}$.}\label{Non-dege-Config}
\end{figure}

The following corollary follows from an inspection of all the possible cases of a non-degenerated $\mathcal{N}_{t_1, t_2}$ as stated in Theorem~\ref{thm-max-config}.
\begin{cor}\label{lem-0-2}
Suppose that $\mathcal{N}_{t_1, t_2}$ is non-degenerated with $\ell(r_i) = \ell(g_l) =1$. Then, $(1)$ $g_i(2)\neq r_l(2)$; $(2)$ there exist a unique $\{g_i, r_j\}$-l.c.s. and a unique $\{g_j, r_l\}$-l.c.s., where $j$ is distinct from both $i$ and $l$; and $(3)$ one of the following $3$ statements holds:
\begin{description}
  \item[$a)$] $g_i(2)=r_j(2)$ and $g_j(2)=r_l(2)$;
  \item[$b)$] $g_i(2)=r_j(2)$ and $g_j(2)=r_l(3)$;
  \item[$c)$] $g_i(3)=r_j(2)$ and $g_j(2)=r_l(2)$.
\end{description}
\end{cor}

We say $\mathcal{N}_{t_1, t_2}$ is of {\em type $1$} if $a)$ of Corollary~\ref{lem-0-2} holds, and of {\em type $2$} otherwise. It is easy to check that in Fig.~\ref{Non-dege-Config}, $(1.1.1), (1.1.2), (1.2.1)$ are of type $1$ and $(1,2,2,1),(1,2,2,2)$ are of type $2$ since they satisfy $b)$. Note that in Fig.~\ref{Non-dege-Config}, if we switch colors of paths and source labels $i$ and $l$, then $(1.1.1), (1.1.2), (1.2.1)$ still satisfy $a)$ but $(1,2,2,1),(1,2,2,2)$ satisfy $c)$ of Corollary~\ref{lem-0-2} instead.

\subsection{A Forbidden Structure}

For any $\mathbf{P}$-path $p$ in $\mathcal{N}_{t_i, t_j}$, let $\ell^{i,j}(p)$ be the number of $\{P_{s_l, t_i}, P_{s_{l'}, t_j}\}$-l.c.s's ($1\leq l, l'\leq 3$) on $p$ and we order such l.c.s's as $p^{i,j}(1) <p^{i,j}(2) < \dots < p^{i,j}(\ell^{i,j}(p))$. Here we remark that the notation $\ell^{i, j}(p)$, $p^{i, j}(\cdot)$ subsume $\ell(p)$, $p(\cdot)$ as the latter two are simply $\ell^{1, 2}(p)$, $p^{1, 2}(\cdot)$, respectively. By Lemma~\ref{lem-crossing}, the following sets are non-empty:
$$
m_{i,j}^i:=\{l: \ell^{i,j}(P_{s_l, t_i})=1\}, \quad m_{i,j}^j:=\{l: \ell^{i,j}(P_{s_l, t_j})=1\}.
$$
Here, let us add that the two subscripts of $m_{i,j}^i$ are interchangeable, more precisely, $m_{i,j}^i=m_{j,i}^i$. In the case $m_{i,j}^i$ contains only one element, e.g., $\mathcal{N}_{t_i, t_j}$ is degenerated, we may write $m_{i,j}^i=l$ instead of $m_{i, j}^i=\{l\}$ for simplicity. For example, for the network  depicted in Fig.~\ref{Minimal Networks}(a), each $\mathcal{N}_{t_i, t_j}$ is non-degenerated and $m_{1,2}^1=1$, $m_{1,2}^2=3$, $m_{1,3}^1=1$, $m_{1,3}^3=3$, $m_{2,3}^2=1$ and $m_{2,3}^3=3$.

\begin{thm}\label{thm-0-3}
There exists no stable network such that $(1)$ $m_{i,j}^i=l$, $m_{i,j}^j=i$; $m_{i,l}^i=l$, $m_{i,l}^l=j$; $m_{j,l}^j=i$, $m_{j,l}^l=j$; $(2)$ there exists a $\{P_{s_j, t_i}, P_{s_l, t_j}, P_{s_i, t_l}\}$-l.c.s, where $i, j, l$ are all distinct from one another.
\end{thm}

\begin{figure}[htbp]
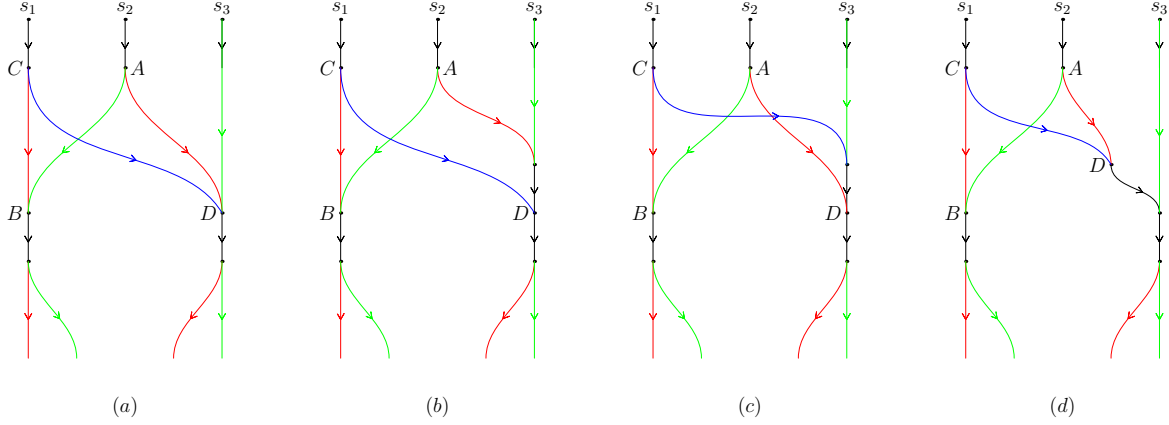

  \centering
  \includegraphics[width=3cm]{fig4.1}~~~~~~~~~~
  \includegraphics[width=3cm]{fig4.2}~~~~~~~~~~
  \includegraphics[width=3cm]{fig4.3}~~~~~~~~~~
  \includegraphics[width=3cm]{fig4.4}~~~~~~~~~~
  \caption{Proof of Case $1)$ of Theorem \ref{thm-0-3}, where we do not show path $b_3$. Note that $r_2^{1,2}(2)=g_3^{1,2}(2)$, $r_2^{1,3}(2)=b_1^{1,3}(2)$, $g_3^{2,3}(2)=b_1^{2,3}(2)$ and there exists a unique $\{r_2, g_3,b_1\}$-l.c.s.}\label{fig-case-1}
\end{figure}

\begin{proof}
Suppose, by way of contradiction, that there exists a stable network $\mathcal N$ such that $(1)$ and $(2)$ hold. Without loss of generality, we assume $i=1$, $j=2$ and $l=3$ and therefore $m_{1,2}^1=m_{1,3}^1=3$, $m_{2,3}^2=m_{1,2}^2=1$, $m_{2,3}^3=m_{1,3}^3=2$, which implies $\ell^{1,2}(r_3)=\ell^{1,3}(r_3)=1$, $\ell^{1,2}(g_1)=\ell^{2,3}(g_1)=1$ and $\ell^{1,3}(b_2)=\ell^{2,3}(b_2)=1$. We consider the following two cases:
\begin{description}
  \item[$1)$] all $\mathcal{N}_{t_1, t_2}$, $\mathcal{N}_{t_1, t_3}$ and $\mathcal{N}_{t_2, t_3}$ are of type $1$;
  \item[$2)$] any of $\mathcal{N}_{t_1, t_2}$, $\mathcal{N}_{t_1, t_3}$ or $\mathcal{N}_{t_2, t_3}$ is of type $2$.
\end{description}

We first prove the theorem for Case $1)$. Consider $\mathcal{N}_{t_1, t_2}$. Since it is of type $1$, by Lemma \ref{lem-0-2}, we have that $r_1^{1,2}(2)=g_2^{1,2}(2)$, $r_2^{1,2}(2)=g_3^{1,2}(2)$. Note that although there are several types of $\{r_2, g_3,b_1\}$-l.c.s., as shown in $(a)-(d)$ of Fig.~\ref{fig-case-1}, our argument in the following however does not depend on the specific type. In the following, we prove $g_2^{1,2}(1, 2)$ and $b_1^{1,3}(1, 2)$ (shown as $[A, B]$ and $[C, D]$, respectively in Fig.~\ref{fig-case-1}) form a $r$-crossing, which will contradict Theorem~\ref{thm-main} and yield the theorem for this case. Towards this goal, we only need to prove the following two statements:
\begin{description}
  \item[$(a)$] $C<B$ and $A<D$;
  \item[$(b)$] $g_2^{1,2}(1, 2)\cap b_1^{1,3}(1, 2)=\emptyset$.
\end{description}

For $(a)$, it is easy to see that either $B\leq C$ or $D\leq A$ will imply that $b_1^{2,3}(2) \subseteq g_2$, which contradicts the fact that $b_1^{2,3}(2)=g_3^{2,3}(2)$ since $\mathcal{N}_{t_2, t_3}$ is of type 1. Hence $(a)$ holds. For $(b)$, it is easy to see that $g_2^{1,2}(1, 2)\cap b_1^{1,3}(1, 2) \neq \emptyset$ also contradicts the fact that $b_1^{2,3}(2)=g_3^{2,3}(2)$. Hence, $(b)$ holds.

Now, we prove the theorem for Case $2)$. Without loss of generality, we suppose $\mathcal{N}_{t_1, t_2}$ is of type $2$. Then, according to Corollary~\ref{lem-0-2}, there are two possible cases. Specifically, in Fig.~\ref{fig-case-2-1} (resp. Fig.~\ref{fig-case-2-2}), $(a)$ satisfies: $r_1^{1,2}(2)=g_2^{1,2}(2)$ and $r_2^{1,2}(2)=g_3^{1,2}(3)$; and  $(b)$ satisfies: $r_1^{1,2}(3)=g_2^{1,2}(2)$ and $r_2^{1,2}(2)=g_3^{1,2}(2)$.

We consider the following two cases:
\begin{description}
  \item[$2.1)$] $b_1^{2,3}(2)\subseteq g_3$;
  \item[$2.2)$] $b_1^{2,3}(2)\subseteq g_2$.
\end{description}

\begin{figure}[htbp]
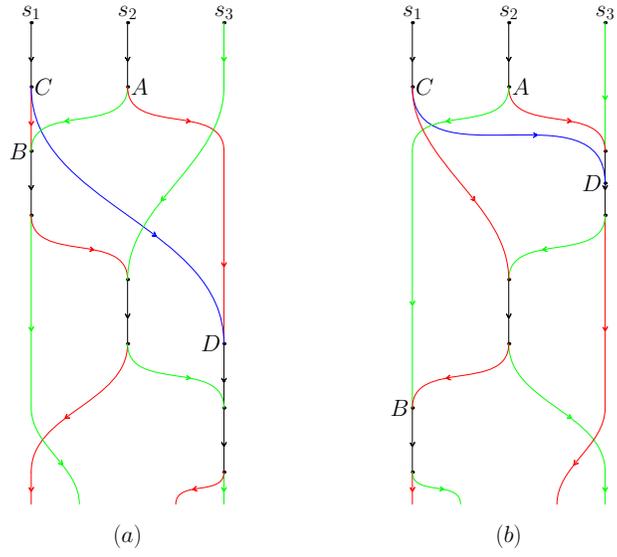

  \centering
  \includegraphics[width=3cm]{fig4.7}~~~~~~~~~~~~~~
  \includegraphics[width=3cm]{fig4.8}
  \caption{Proof of Case $2.1)$ of Theorem \ref{thm-0-3}, where we do not show path $b_3$.} \label{fig-case-2-1}
\end{figure}

For Case $2.1)$ (see Fig.~\ref{fig-case-2-1}), the proof is similar to that of Case $1)$. In this case, we can have $g_2^{1,2}(1, 2)$ and $b_1^{1,3}(1, 2)$ (shown as $[A, B]$ and $[C, D]$, respectively in Fig.~\ref{fig-case-2-1}) form a $r$-crossing, which contradicts the stability of the network.

For Case $2.2)$, since $b_1^{2,3}(2)\neq g_2^{2,3}(2)$ and $b_1^{2,3}(2)\not \subseteq g_3$, we have, by Lemma \ref{lem-0-2}, $b_1^{2,3}(2)=g_2^{2,3}(3)$ and $b_3^{2,3}(2)=g_2^{2,3}(2)$, as shown in Fig.~\ref{fig-case-2-2}. Let $A:=head(r_2^{1,2}(1))$, $B:=tail(r_2^{1,2}(2))$, $C:=head(b_3^{2,3}(1))$, $D:=head(b_3^{2,3}(2))$. Noticing that either $B\leq C$ or $D\leq A$ or $r_2^{1,2}(1, 2)\cap b_3^{2,3}(1, 2)\neq\emptyset$ will imply that $r_2^{1,3}(2)=b_3^{1,3}(2)$, which however, is impossible according to $(1)$ of Corollary~\ref{lem-0-2}. Hence, we have

\begin{figure}[htbp]
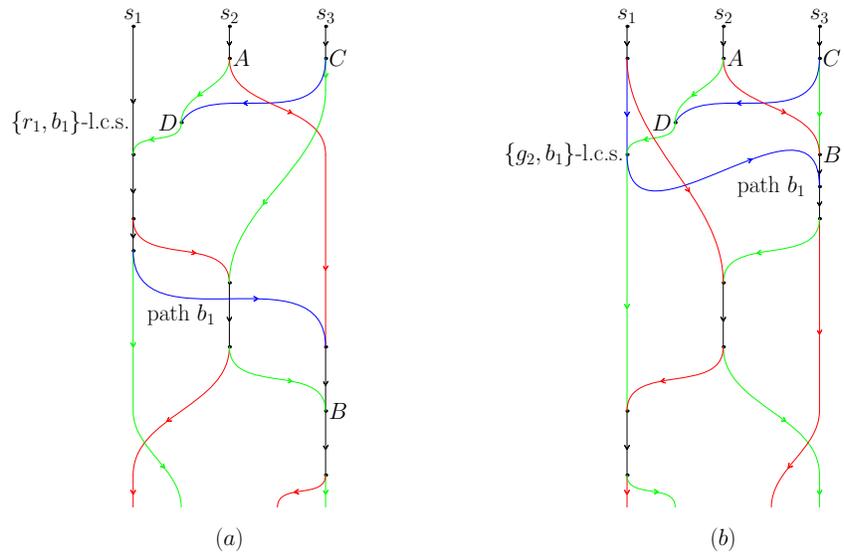

  \centering
  \includegraphics[width=4.5cm]{fig4.9}~~~~~~~~~~~~~~
  \includegraphics[width=4.5cm]{fig4.10}
  \caption{Proof of Case $2.2)$ of Theorem \ref{thm-0-3}.}\label{fig-case-2-2}
\end{figure}

\begin{description}
  \item[$(a)$] $C<B$ and $A<D$;
  \item[$(b)$] $r_2^{1,2}(1, 2)\cap b_3^{2,3}(1, 2)=\emptyset$.
\end{description}
By definition, $r_2^{1,2}(1, 2)$ and $b_3^{2,3}(1, 2)$  form a $g$-crossing, a contradiction that leads to the theorem for this case.

The proof is then complete by combining all the discussions above.
\end{proof}

\section{Main Result}

In this section, we state and prove our main result. Throughout this section, we again assume that $\mathcal{N}$ is a stable $3$-pair network.

The following seemingly trivial lemma is a key tool for us to determine $\mathcal S_{\mathcal N}$ throughout our treatment.
\begin{lem}\label{lem-basic}
If there are no $\{P_{s_{i_1},t_{j_1}}, P_{s_{i_2},t_{j_2}}\}$-l.c.s. within $\mathcal N$, then $\{(i_1,j_1),(i_2,j_2) \} \nsubseteq s$ for any $s\in \mathcal S_{\mathcal N}$.
\end{lem}

The following lemma is useful.
\begin{lem}\label{lem-0-1}
If there exists $h$ such that $h\in m_{i,j}^{l_1}\cap m_{i,l}^{l_2}\cap m_{j,l}^{l_3}$ for some feasible $l_1, l_2, l_3$ and distinct $i, j, l$, then either $\{(h,i),(h,j)\}\notin\mathcal S_{\mathcal N}$ or $\{(h,i),(h,l)\}\notin\mathcal S_{\mathcal N}$.
\end{lem}

\begin{figure}[htbp]
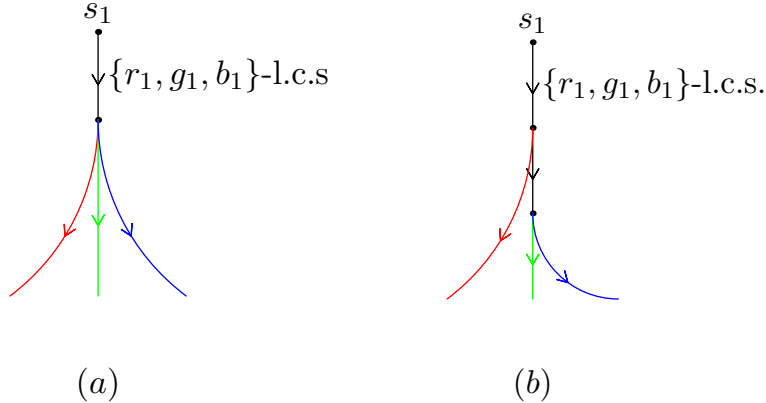

  \centering
  \includegraphics[width=4.3cm]{fig1.12}~~~~~~~~~~
  \includegraphics[width=4.3cm]{fig1.13}~~~~~~~~~~
  \caption{Proof of Lemma \ref{lem-0-1}. In $(a)$, the unique $\{r_1,g_1, b_1\}$-l.c.s. is also the unique $\{r_1,g_1\}$-l.c.s., $\{g_1, b_1\}$-l.c.s. and $\{r_1, b_1\}$-l.c.s.. In $(b)$, the unique $\{r_1,g_1, b_1\}$-l.c.s. is also the unique $\{r_1,g_1\}$-l.c.s. and $\{r_1, b_1\}$-l.c.s.. }\label{Same source}
\end{figure}

\begin{proof}
Without loss of generality, we assume $h=1$. Then, there exist a unique $\{r_1,g_1\}$-l.c.s., a unique $\{r_1, b_1\}$-l.c.s. and a unique $\{g_1, b_1\}$-l.c.s. with a same tail $s_1$. If all of them have a same head, as shown in $(a)$ of Fig.~\ref{Same source}, then $\{(1,1),(1,2),(1,3)\}\in \mathcal S_{\mathcal N}$ and none of $\{(1,1),(1,2)\}$, $\{(1,1),(1,3)\}$, $\{(1,2),(1,3)\}$ belongs to $\mathcal S_{\mathcal N}$; otherwise two of them share a same head, then $\{(1,1),(1,2),(1,3)\}\in \mathcal S_{\mathcal N}$ and at most one of $\{(1,1),(1,2)\}$, $\{(1,1),(1,3)\}$, $\{(1,2),(1,3)\}$ belongs to $\mathcal S_{\mathcal N}$ (for example, $(b)$ shows the case $\{(1,2),(1,3)\}\in \mathcal S_{\mathcal N}$). Hence, the result holds for both cases, which completes the proof.
\end{proof}

We also need the following lemma.

\begin{lem}\label{lem-1-2}
Let
$$
\mathcal C=\left(\left(
        \begin{array}{ccc}
          \frac{1}{2} & \frac{1}{4} &\frac{1}{4} \\
          \frac{1}{4}& \frac{-1}{4} & 0 \\
          \frac{1}{4} & 0 & \frac{-1}{4} \\
        \end{array}
      \right),
      \left(
        \begin{array}{ccc}
          \frac{-1}{4} & \frac{1}{4} & 0 \\
          \frac{1}{4} & \frac{1}{2} & \frac{1}{4} \\
          0 & \frac{1}{4} & \frac{-1}{4} \\
        \end{array}
      \right),
      \left(
        \begin{array}{ccc}
          \frac{-1}{4} & 0 & \frac{1}{4} \\
          0 & \frac{-1}{4} & \frac{1}{4} \\
          \frac{1}{4} & \frac{1}{4} & \frac{1}{2} \\
        \end{array}
      \right)
\right).
$$
Then for any $s\in \mathcal S_k$, $g_s(\mathcal C)>1$ if and only if $\alpha(s)=3$ and $\gamma(s)=0$.
\end{lem}

\begin{proof}
The result can be obtained by considering the following cases:
\begin{description}
  \item[$1)$] $\gamma(s)=0$. In this case, it is easy to see that
  $$
  g_s(\mathcal C)=\frac{1}{4}\sum_{i=1}^3 m_{Ind_s}(i)=2\alpha(s)=\left\{
   \begin{array}{ll}
   \frac{1}{2}, & \hbox{$\alpha(s)=1$;} \\
   1, & \hbox{$\alpha(s)=2$;} \\
   \frac{3}{2}, & \hbox{$\alpha(s)=3$.}
   \end{array}
   \right.
  $$
  \item[$2)$] $\gamma(s)=1$. In this case, it is easy to check that
  $$
  g_s(\mathcal C)=\left\{
   \begin{array}{ll}
   \frac{1}{2}+\frac{1}{4}+\frac{1}{4}=1, & \hbox{$\alpha(s)=1$;} \\
   \frac{3}{4}+\frac{1}{4}=1, & \hbox{$s=\{(i,i),(i,j)\}$, where $i\neq j$;} \\
   \frac{1}{2}, & \hbox{$s=\{(i,i),(k,j)\}$, where $i,j,k$ are distinct;} \\
   1, & \hbox{$\alpha(s)=3$.}
   \end{array}
   \right.
  $$
  \item[$3)$] $\gamma(s)=2$. In this case, it is easy to check that $g_s(\mathcal C)=1$.
  \item[$4)$] $\gamma(s)=3$. In this case, obviously, $g_s(\mathcal C)=0$.
\end{description}
\end{proof}

We are now ready for our main result.
\begin{thm}
Each stable $3$-pair network has a linear routing solution.
\end{thm}

\begin{proof}
For the stable $3$-pair network $\mathcal{N}$, we consider the following two cases:
\begin{description}
  \item[$1)$] there exist distinct $i,j,l\in\{1,2,3\}$, $m_{i,j}^i\cap\{i,j\}\neq\emptyset$ and $m_{i,l}^i\cap\{i,l\}\neq\emptyset$;
  \item[$2)$] for any distinct $i,j,l\in\{1,2,3\}$, either $m_{i,j}^i=l$ or $m_{i,l}^i=j$.
\end{description}

For Case $1)$, we have the following subcases:
\begin{description}
  \item[$1.1)$] $i\in m_{i,j}^i\cap m_{i,l}^i$;
  \item[$1.2)$] $j\in m_{i,j}^i$ and $l\in m_{i,l}^i$;
  \item[$1.3)$] $i\in m_{i,j}^i$ and $l\in m_{i,l}^i$.
\end{description}

In the following, without loss of generality, we assume $i=1$, $j=2$ and $l=3$.

For Case $1.1)$, if $1\in m_{1,2}^1\cap m_{1,3}^1$, then $r_1$ is disjoint from $\mathcal N':=\{g_2, g_3, b_2, b_3\}$, which is a stable (hence extra strongly reachable) $2$-pair network. By~\cite{CH17}, $\mathcal N'$ always has a linear routing solution
$$
\left(\left(
  \begin{array}{cc}
    \frac{3}{4} & \frac{1}{4} \\
    \frac{1}{4} & \frac{-1}{4} \\
  \end{array}
\right),
\left(
  \begin{array}{cc}
    \frac{-1}{4} & \frac{1}{4} \\
    \frac{1}{4} & \frac{3}{4} \\
  \end{array}
\right)
 \right).
$$
Hence, $\mathcal{N}$ has the following linear routing solution:
$$
\left(\left(
        \begin{array}{ccc}
          1 & 0 & 0 \\
          0 & 0 & 0 \\
          0 & 0 & 0 \\
        \end{array}
      \right),
      \left(
        \begin{array}{ccc}
          0 & 0 & 0 \\
          0 & \frac{3}{4} & \frac{1}{4} \\
          0 & \frac{1}{4} & -\frac{1}{4} \\
        \end{array}
      \right),
      \left(
        \begin{array}{ccc}
          0 & 0 & 0 \\
          0 & -\frac{1}{4} & \frac{1}{4} \\
          0 & \frac{1}{4} & \frac{3}{4} \\
        \end{array}
      \right)
\right).
$$

For Case $1.2)$, consider all $s\in \mathcal S_{\mathcal N}\subseteq \mathcal S_3$ such that $\alpha(s)=3$. Let $s=\{(l_1,1),(l_2, 2), (l_3, 3)\}$. If $l_1=1$, then obviously $\gamma(s)\neq 0$; if $l_1=2$, then since $2\in m_{1,2}^1$, we have $l_2=2$ and hence $\gamma(s)\neq 0$; and if $l_1=3$, since $3\in m_{1,3}^1$, we have $l_3=3$ and hence $\gamma(s)\neq 0$. Thus, for any $s\in \mathcal S_{\mathcal N}$ such that $\alpha(s)=3$, we have $\gamma(s)\neq 0$.
By Lemma \ref{lem-1-2},
$$
\left(\left(
        \begin{array}{ccc}
          \frac{1}{2} & \frac{1}{4} &\frac{1}{4} \\
          \frac{1}{4}& \frac{-1}{4} & 0 \\
          \frac{1}{4} & 0 & \frac{-1}{4} \\
        \end{array}
      \right),
      \left(
        \begin{array}{ccc}
          \frac{-1}{4} & \frac{1}{4} & 0 \\
          \frac{1}{4} & \frac{1}{2} & \frac{1}{4} \\
          0 & \frac{1}{4} & \frac{-1}{4} \\
        \end{array}
      \right),
      \left(
        \begin{array}{ccc}
          \frac{-1}{4} & 0 & \frac{1}{4} \\
          0 & \frac{-1}{4} & \frac{1}{4} \\
          \frac{1}{4} & \frac{1}{4} & \frac{1}{2} \\
        \end{array}
      \right)
\right).
$$
is a linear solution of $\mathcal N$.

For Case $1.3)$, since $1\in m_{1,2}^1$ and $3\in m_{1,3}^1$, by Lemma \ref{lem-basic}, we have
\begin{equation*}
\begin{split}
  \mathcal S_{\mathcal N}\subseteq \mathcal S:
  &=\{\{(i,j)\}: 1\leq i,j\leq 3\} \\
  &\cup \{\{(i,1),(j,2)\}: i=2,3;j=1,2,3\}\cup\{\{(1,1),(1,2)\}\}\\
  &\cup \{\{(i,1),(l,3)\}: i=1,2;l=1,2,3\}\cup\{\{(3,1),(3,3)\}\}\\
  &\cup \{\{(j,2),(l,3)\}: j,l=1,2,3\}\\
  &\cup \{\{(1,1),(1,2),(l,3)\}:l=1,2,3\}\cup\{\{(2,1),(j,2),(l,3)\}:j,l=1,2,3\}\\
  &\cup\{\{(3,1),(j,2),(3,3)\}:j=1,2,3\}.
\end{split}
\end{equation*}
Let
$$
\mathcal C=\left(\left(
        \begin{array}{ccc}
          \frac{3}{4} & 0 &\frac{1}{4} \\
          0& 0 & 0 \\
         \frac{1}{4} & 0 & \frac{-1}{4} \\
        \end{array}
      \right),
      \left(
        \begin{array}{ccc}
          0 & 0 & 0 \\
          0 & \frac{3}{4} & \frac{1}{4} \\
          0 & \frac{1}{4} & \frac{-1}{4} \\
        \end{array}
      \right),
      \left(
        \begin{array}{ccc}
          \frac{-1}{4} & 0 & \frac{1}{4} \\
          0 & \frac{-1}{4} & \frac{1}{4} \\
          \frac{1}{4} & \frac{1}{4} & \frac{1}{2} \\
        \end{array}
      \right)
\right).
$$
Through straightforward computations, one can verify that for any $s\in \mathcal S$, $g_s(\mathcal C)\leq 1$. Hence, by Theorem~\ref{thm-comput}, $\mathcal C$ is a linear solution of $\mathcal N$, which completes the proof of Case $1)$.

For Case $2)$, without loss of generality, we assume $m_{i,j}^i=l$. Note that by Lemma \ref{lem-crossing}, if  $m_{i,j}^i=l$, then $m_{i,j}^j\cap\{i,j\}\neq \emptyset$, which further implies $m_{l,j}^j=i$ by the assumption of this case. Hence, by Lemma \ref{lem-crossing}, we have $m_{l,j}^l\cap\{l,j\}\neq\emptyset$, which implies $m_{l,i}^l=j$, again by the assumption of this case, and further implies $m_{l,i}^i\cap\{i,l\}\neq\emptyset$ by Lemma \ref{lem-crossing}. Finally, we have $l\in m_{i,j}^i$, $j\in m_{i,l}^l$ and $i\in m_{j,l}^j$. Consider the following subcases:
\begin{description}
  \item[$2.1)$]   $j\in m_{i,j}^j$, $i\in m_{i,l}^i$ and $l\in m_{j,l}^l$;
  \item[$2.2)$]   $i\in m_{i,j}^j$, $i\in m_{i,l}^i$ and $l\in m_{j,l}^l$;
  \item[$2.2')$]  $j\in m_{i,j}^j$, $i\in m_{i,l}^i$ and $j\in m_{j,l}^l$;
  \item[$2.2'')$] $j\in m_{i,j}^j$, $l\in m_{i,l}^i$ and $l\in m_{j,l}^l$;
  \item[$2.3)$]   $i\in m_{i,j}^j$, $i\in m_{i,l}^i$ and $j\in m_{j,l}^l$;
  \item[$2.3')$]  $i\in m_{i,j}^j$, $l\in m_{i,l}^i$ and $l\in m_{j,l}^l$;
  \item[$2.3'')$] $j\in m_{i,j}^j$, $l\in m_{i,l}^i$ and $j\in m_{j,l}^l$;
  \item[$2.4)$]   $i\in m_{i,j}^j$, $l\in m_{i,l}^i$ and $j\in m_{j,l}^l$.
\end{description}
It is easy to check that Cases $2.2')$ and $2.2'')$ can be obtained form Case $2.2)$ (resp. Cases $2.3')$ and $2.3'')$ can be obtained form Case $2.3)$) by the relabelling: $i\mapsto j$, $j\mapsto l$, $l\mapsto i$ and the relabelling: $i\mapsto l$, $j\mapsto i$, $l\mapsto j$, respectively. So, in the following, we only need to consider Cases $2.1), 2.2), 2.3), 2.4)$.

For Case $2.1)$, without loss of generality, we assume $i=1$, $j=2$ and $l=3$ and thus $2\in m_{1,2}^2$, $1\in m_{1,3}^1$ and $3\in m_{2,3}^2$. Hence, paths $r_1$, $g_2$ and $b_3$ are pairwise disjoint and the network has a linear routing solution
$$
\left(\left(
        \begin{array}{ccc}
          1 & 0 & 0 \\
          0 & 0 & 0 \\
          0 & 0 & 0 \\
        \end{array}
      \right),
      \left(
        \begin{array}{ccc}
          0 & 0 & 0 \\
          0 & 1 & 0 \\
          0 & 0 & 0 \\
        \end{array}
      \right),
      \left(
        \begin{array}{ccc}
          0 & 0 & 0 \\
          0 & 0 & 0 \\
          0 & 0 & 1 \\
        \end{array}
      \right)
\right).
$$

For Case $2.2)$, without loss of generality, we assume $i=1$, $j=2$ and $l=3$ and thus $1\in m_{1,2}^2\cap m_{1,3}^1\cap m_{2,3}^2$; $2\in m_{1,3}^3$ and $3\in m_{1,2}^1\cap m_{2,3}^3$. By Lemma \ref{lem-basic}, we have
\begin{equation*}
\begin{split}
  \mathcal S_{\mathcal N}\subseteq \mathcal S: &=\{\{(i,j)\}: 1\leq i,j\leq 3\}\\
  &\cup\{\{(i,1),(j,2)\}: i=1,2;j=2,3\}\cup\{\{(1,1),(1,2)\}, \{(3,1),(3,2)\}\}\\
  &\cup\{\{(i,1),(l,3)\}: i=2,3;l=1,3\}\cup\{\{(1,1),(1,3)\},\{(2,1),(2,3)\}\}\\
  &\cup\{\{(j,2),(l,3)\}: j=2,3;l=1,2\}\cup\{\{(1,2),(1,3)\},\{(3,2),(3,3)\}\}\\
  &\cup\{\{(1,1),(j,2),(1,3)\}:j=1,2,3\}\cup\{\{(2,1),(2,2),(l,3)\}:l=1,2\}\\
  &\cup\{\{(2,1),(3,2),(l,3)\}: l=1,2,3\}\cup\{\{(3,1),(3,2),(l,3)\}:l=1,3\}.
\end{split}
\end{equation*}
Then, by Lemma \ref{lem-0-1}, we have the following two subcases:

If $\{(1,1),(1,2)\}\notin\mathcal S_{\mathcal N}$, then, $\mathcal S_{\mathcal N}\subseteq\mathcal S\setminus\{(1,1),(1,2)\}$ and by Theorem~\ref{thm-comput}, one can check that
$$
\left(\left(
        \begin{array}{ccc}
          \frac{8}{14} & \frac{7}{14} &\frac{-1}{14} \\
          \frac{3}{14}& \frac{-5}{14} & \frac{2}{14} \\
          \frac{3}{14} & \frac{-2}{14} & \frac{-1}{14} \\
        \end{array}
      \right),
      \left(
        \begin{array}{ccc}
          \frac{-3}{14} & \frac{7}{14} & \frac{-4}{14} \\
          \frac{3}{14} & \frac{7}{14} & \frac{4}{14} \\
          0 & 0 & 0 \\
        \end{array}
      \right),
      \left(
        \begin{array}{ccc}
          \frac{-3}{14} & 0 & \frac{3}{14} \\
          0 & \frac{-2}{14} & \frac{2}{14} \\
          \frac{3}{14} & \frac{2}{14} & \frac{9}{14} \\
        \end{array}
      \right)
\right)
$$
is a linear solution.

If $\{(1,1),(1,3)\}\notin\mathcal S_{\mathcal N}$, then $\mathcal S_{\mathcal N}\subseteq\mathcal S\setminus\{\{(1,1),(j,2),(1,3)\}:j=2,3\}$ and by Theorem~\ref{thm-comput}, one can check that
$$
\left(\left(
        \begin{array}{ccc}
          \frac{6}{12} & \frac{3}{12} &\frac{3}{12} \\
          \frac{3}{12}& \frac{-3}{12} & 0\\
          \frac{3}{12} & 0 & \frac{-3}{12} \\
        \end{array}
      \right),
      \left(
        \begin{array}{ccc}
          \frac{-3}{12} & \frac{4}{12} & \frac{-1}{12} \\
          \frac{3}{12} & \frac{7}{12} & \frac{2}{12} \\
          0 & \frac{1}{12} & \frac{-1}{12} \\
        \end{array}
      \right),
      \left(
        \begin{array}{ccc}
          \frac{-3}{12} & \frac{1}{12} & \frac{2}{12} \\
          0 & \frac{-2}{12} & \frac{2}{12} \\
          \frac{3}{12} & \frac{1}{12} & \frac{8}{12} \\
        \end{array}
      \right)
\right)
$$
is a linear solution, which proves the theorem for Case $2.2)$.

For Case $2.3)$, without loss of generality, we assume $i=1$, $j=2$ and $l=3$. It can be readily verified that $1\in m_{1,2}^2\cap m_{1,3}^1\cap m_{2,3}^2$; $2\in m_{1,3}^3\cap m_{2,3}^3$ and $3\in m_{1,2}^1$. By Lemma \ref{lem-basic},
\begin{equation*}
\begin{split}
  \mathcal S_{\mathcal N}\subseteq \mathcal S:&=\{\{(i,j)\}: 1\leq i,j\leq 3\}\\
  &\cup\{\{(i,1),(j,2)\}: i=1,2;j=2,3\}\cup\{\{(1,1),(1,2)\}, \{(3,1),(3,2)\}\}\\
  &\cup\{\{(i,1),(l,3)\}: i=2,3;l=1,3\}\cup\{\{(1,1),(1,3)\},\{(2,1),(2,3)\}\}\\
  &\cup\{\{(j,2),(l,3)\}: j=2,3;l=1,3\}\cup\{\{(1,2),(1,3)\},\{(2,2),(2,3)\}\}\\
  &\cup\{\{(1,1),(j,2),(1,3)\}:j=1,2,3\}\cup\{\{(2,1),(2,2),(l,3)\}:l=1,2,3\}\\
  &\cup\{\{(2,1),(3,2),(l,3)\}: l=1,3\}\cup\{\{(3,1),(3,2),(l,3)\}:l=1,3\}.
\end{split}
\end{equation*}
By Lemma \ref{lem-0-1}, we have the following two subcases:

If $\{(1,1),(1,2)\}\notin\mathcal S_{\mathcal N}$, then, $\mathcal S_{\mathcal N}\subseteq\mathcal S\setminus\{(1,1),(1,2)\}$ and  by Theorem~\ref{thm-comput}, one can check that
$$
\left(\left(
        \begin{array}{ccc}
          \frac{4}{8} & \frac{3}{8} &\frac{1}{8} \\
          \frac{2}{8}& \frac{-2}{8} & 0 \\
          \frac{2}{8} & \frac{-1}{8} & \frac{-1}{8} \\
        \end{array}
      \right),
      \left(
        \begin{array}{ccc}
          \frac{-2}{8} & \frac{3}{8} & \frac{-1}{8} \\
          \frac{2}{8} & \frac{4}{8} & \frac{2}{8} \\
          0           & \frac{1}{8}  & \frac{-1}{8}  \\
        \end{array}
      \right),
      \left(
        \begin{array}{ccc}
          \frac{-2}{8} & 0 & \frac{2}{8} \\
          0 & \frac{-2}{8} & \frac{2}{8} \\
          \frac{2}{8} & \frac{2}{8} & \frac{4}{8} \\
        \end{array}
      \right)
\right)
$$
is a linear solution.

If $\{(1,1),(1,3)\}\notin\mathcal S_{\mathcal N}$, then, $\mathcal S_{\mathcal N}\subseteq\mathcal S \setminus\{\{(1,1),(1,3)\}\}\setminus\{\{(1,1),(j,2),(1,3)\}:j=2,3\}$ and by Theorem~\ref{thm-comput}, one can check that
$$
\left(\left(
        \begin{array}{ccc}
          \frac{4}{6} & \frac{1}{6} &\frac{1}{6} \\
          \frac{1}{6}& \frac{-1}{6} & 0\\
          \frac{1}{6} & 0 & \frac{-1}{6} \\
        \end{array}
      \right),
      \left(
        \begin{array}{ccc}
          \frac{-2}{6} & \frac{3}{6} & \frac{-1}{6} \\
          \frac{2}{6} & \frac{2}{6} & \frac{2}{6} \\
          0           & \frac{1}{6} & \frac{-1}{6} \\
        \end{array}
      \right),
      \left(
        \begin{array}{ccc}
              0      &    0         &    0       \\
          \frac{-1}{6}& \frac{-1}{6} & \frac{2}{6} \\
          \frac{1}{6} & \frac{1}{6} & \frac{4}{6} \\
        \end{array}
      \right)
\right)
$$
is a linear solution, which proves the theorem for Case $2.3)$.

For Case $2.4)$, if one of $\mathcal N_{t_1, t_2}$, $\mathcal N_{t_1, t_3}$ and $\mathcal N_{t_2, t_3}$ is degenerated, then $\mathcal N$ has a linear solution by previous cases. So, we assume all of them are non-degenerated and without loss of generality $i=1$, $j=2$ and $l=3$. Hence, $m_{1,2}^1=3$, $m_{1,2}^2=1$; $m_{1,3}^1=3$, $m_{1,3}^3=2$; and $m_{2,3}^2=1$, $m_{2,3}^3=2$. In the following, consider $s\in\mathcal S_{\mathcal N}\subseteq \mathcal S_3$ such that $\alpha(s)=3$. Let $s=\{(l_1,1),(l_2,2),(l_3,3)\}$. If $l_1=3$, then since $m_{1,3}^1=3$,  we have $l_3=3$; if $l_2=1$, then since $m_{1,2}^2=1$, we have $l_1=1$; if $l_3=2$, then since $m_{2,3}^3=2$, we have $l_2=2$. Hence, $\gamma(s)=0$ only if $s=\{(2,1),(3,2),(1,3)\}$, which however, is impossible by Theorem~\ref{thm-0-3}.

Hence, by Lemma~\ref{lem-1-2},
$$
\left(\left(
        \begin{array}{ccc}
          \frac{1}{2} & \frac{1}{4} &\frac{1}{4} \\
          \frac{1}{4}& \frac{-1}{4} & 0 \\
          \frac{1}{4} & 0 & \frac{-1}{4} \\
        \end{array}
      \right),
      \left(
        \begin{array}{ccc}
          \frac{-1}{4} & \frac{1}{4} & 0 \\
          \frac{1}{4} & \frac{1}{2} & \frac{1}{4} \\
          0 & \frac{1}{4} & \frac{-1}{4} \\
        \end{array}
      \right),
      \left(
        \begin{array}{ccc}
          \frac{-1}{4} & 0 & \frac{1}{4} \\
          0 & \frac{-1}{4} & \frac{1}{4} \\
          \frac{1}{4} & \frac{1}{4} & \frac{1}{2} \\
        \end{array}
      \right)
\right).
$$
is a linear routing solution of $\mathcal N$, which completes the proof.
\end{proof}

\section{Conclusions and Future Work}

We have settled in this work the Langberg-M\'{e}dard multiple unicast conjecture for stable $3$-pair networks. The conjecture in more general settings, e.g., unstable $3$-pair networks and even more general $k$-pair networks, is currently under investigation.


\begin{thebibliography}{}

\bibitem{CH15}
K.~Cai and G.~Han, ``On network coding advantage for multiple unicast networks,'' in \emph{Proc. ISIT}, 2015.

\bibitem{CH16}
K.~Cai and G.~Han, ``Coding advantage in communications among peers,'' in \emph{Proc. ISIT}, 2016.

\bibitem{CH17}
K.~Cai and G.~Han, ``On the Langberg-M\'{e}dard multiple unicast conjecture,'' \emph{Journal of Combinatorial Optimization}, vol. 34, no. 4, pp. 1114-1132, 2017.

\bibitem{CH18}
K.~Cai and G.~Han, ``On the Langberg-M\'{e}dard $k$-unicast conjecture with $k=3,4$,'' in \emph{Proc. ISIT}, 2018.


\bibitem{CLRS09}
Thomas H. Cormen, Charles E. Leiserson, Ronald L. Rivest, and Clifford Stein, ``Introduction to
algorithms,'' MIT Press, Cambridge, MA, third edition, 2009.

\bibitem{Han09}
G.~Han, ``Menger's paths with minimum mergings,'' in \emph{Proc. ITW}, 2009.

\bibitem{Harv06}
N.~Harvey, R.~Kleinberg and A.~Lehman, ``On the capacity of information networks,'' \emph{IEEE Trans. Inf. Theory}, vol. 52, no. 6, pp. 2345-2364, Jun. 2006.

\bibitem{Jain06}
K.~Jain, V.~V.~Vazirani and G.~Yuval, ``On the capacity of multiple unicast sessions in undirected graphs,'' \emph{IEEE/ACM Trans. Networking}, vol. 14, pp. 2805-2809, Jun. 2006.

\bibitem{Langberg06}
M.~Langberg, A. Sprintson and J. Bruck, ``The encoding complexity of network coding,'' \emph{IEEE Trans. Inf. Theory}, vol. 52, no. 6, pp. 2386-2397, Jun. 2006.


\bibitem{Langberg09}
M.~Langberg and M.~M$\acute{e}$dard, ``On the multiple unicast network coding conjecture,'' in \emph{Proc. 47th Annual Allerton}, 2009.

\bibitem{Li042}
Z.~Li and B.~Li, ``Network coding: The case of multiple unicast sessions,'' in \emph{Proc. 42nd Annual Allerton}, 2004.

\bibitem{Zongpeng12}
Z.~Li and C.~Wu, ``Space information flow: multiple unicast,'' in \emph{Proc. ISIT 2012}, July 1-6, 2012.

\bibitem{Schrijver03}
A.~Schrijver, ``{\em Combinatorial Optimization},'' Springer-Verlag, 2003.

\bibitem{Xiahou12}
T.~Xiahou, C.~Wu, J.~Huang and Z.~Li, ``A geometric framework for investigating the multiple unicast network coding conjecture,'' in \emph{Proc. NetCod}, 2012.

\bibitem{Yang14}
Y.~Yang, X.~Yin, X.~Chen, Y.~Yang and Z.~Li, ``A note on the multiple-unicast network coding conjecture,'' \emph{IEEE Communications Letters}, vol. 18, no. 5, pp. 869-872, May 2014.

\bibitem{Yeung06}
R. Yeung, S.-Y. Li, and N. Cai, {\em Network Coding Theory (Foundations and Trends in Communications and Information Theory)}, Now Publishers Inc., Hanover, MA, USA, 2006.




\end{thebibliography}
\end{document}